%% file: root.tex
\documentclass[letterpaper, 10 pt, conference]{ieeeconf}  % Comment this line out if you need a4paper

\IEEEoverridecommandlockouts                              % This command is only needed if
\usepackage[utf8]{inputenc}
\usepackage[T1]{fontenc}
\usepackage[hidelinks]{hyperref}
\usepackage{url}
\usepackage{ifthen}
\usepackage{cite}
\usepackage[cmex10]{amsmath} % Use the [cmex10] option to ensure complicance
\usepackage{amsfonts,amssymb}
\usepackage{graphicx}
\usepackage{color}
\usepackage{comment}

\usepackage{tikz,pgfplots}
\pgfplotsset{width=7cm,compat=1.3}

\usepackage[justification=centering]{caption}
\usepackage[ruled,vlined]{algorithm2e}
\usepackage{algorithmic}
\usepackage{textcomp}
\usepackage{bbm}

\usepackage[arxiv]{optional}  % arxiv/submission
\usepackage{subcaption}

\def\BibTeX{{\rm B\kern-.05em{\sc i\kern-.025em b}\kern-.08em
		T\kern-.1667em\lower.7ex\hbox{E}\kern-.125emX}}

\newcommand{\cI}{\mathcal{I}}

\newtheorem{theorem}{{Theorem}}
\newtheorem{lemma}[theorem]{{Lemma}}

\newtheorem{IEEEproof}{{Proof}}

\allowdisplaybreaks[1]

\def\cX{\mathcal{X}}

\newcommand{\mdef}{\stackrel{\text{\tiny def}}{=}}
\newcommand{\E}{\mathbb{E}}

\def\P{\mathbb{P}}

\usepackage{mathtools}
\def\P{\mathbb{P}}

\title{Joint Information and Mechanism Design for Queues with Heterogeneous Users}

\author{Nasimeh Heydaribeni and Achilleas Anastasopoulos%
	\thanks{This work was supported in part by NSF Grant ECCS-1608361.}% <-this % stops a space
	\thanks{The authors are with the Department of Electrical Engineering and Computer Science, University of Michigan, Ann Arbor, MI, 48105 USA {\tt\small {heydari,anastas}@umich.edu}}
}

\begin{document}

	\maketitle
	\thispagestyle{empty}
	\pagestyle{empty}
	
	%\author{\IEEEauthorblockN{Nasimeh Heydaribeni}
	%		\IEEEauthorblockA{\textit{Department of Electrical and Computer Engineering} \\
	%			\textit{University of Michigan, Ann Arbor, MI, USA} \\
	%			heydari@umich.edu}
	%		\and
	%		\IEEEauthorblockN{Achilleas Anastasopoulos}
	%		\IEEEauthorblockA{\textit{Department of Electrical and Computer Engineering} \\
	%			\textit{University of Michigan, Ann Arbor, MI, USA} \\
	%			anastas@umich.edu}

	\maketitle
	
	\begin{abstract}
		We consider a queue with an unobservable backlog by the incoming users. There is an information designer that observes the queue backlog and makes recommendations to the users arriving at the queue whether to join or not to join the queue. The arriving users have payoff relevant private types. The users, upon arrival, send a message, that is supposed to be their type, to the information designer if they are willing to hear a recommendation. The information designer then creates a recommendation for that specific type of user. The users have to pay a tax in exchange for the information they receive. In this setting, the information designer has two types of commitments. The first commitment is the recommendation policy and the second commitment is the tax function.
		We combine mechanism design and information design to study a queuing system with heterogeneous users.
		In this setting, the information designer is a sender of the information in the information design aspect and a receiver in the mechanism design aspect of the model.
		We formulate an optimization problem that characterizes the solution of the joint design problem.
		We characterize the tax functions and provide structural results for the recommendation policy of the information designer.
	\end{abstract}
	
	%	\begin{keywords}
	%linear quadratic Gaussian (LQG) games, perfect Bayesian  equilibrium (PBE), dynamic games, asymmetric information.
	%	\end{keywords}

	%%%%%%%%%%%%%%%%%%%%%%%%%%%%%%%%%%%%%%%%%%%%%%%%%%%%%%%%%%%%%%%%%%%%%%%
	%%%%%%%%%%%%%%%%%%%%%%%%%%%%%%%%%%%%%%%%%%%%%%%%%%%%%%%%%%%%%%%%%%%%%%%
	%%%%%%%%%%%%%%%%%%%%%%%%%%%%%%%%%%%%%%%%%%%%%%%%%%%%%%%%%%%%%%%%%%%%%%%
	%%%%%%%%%%%%%%%%%%%%%%%%%%%%%%%%%%%%%%%%%%%%%%%%%%%%%%%%%%%%%%%%%%%%%%%
	\section{Introduction}
	Decentralized information is an important and inevitable aspect of today's systems. Each  agent in a system can own some information that others might be interested to know. On the other hand, agents usually act strategically and might not be willing to share their information with others. Therefore, incentives have to be put in place to motivate agents to reveal some parts of their information. There are two main approaches, mechanism design and information design,  where the sharing of information between agents and their incentives of doing so is studied.
	
	In mechanism design
	%	\textcolor{red}{[AA: add a more extensive literature from other people in the area that will probably be your reviewers as well...]}
	\cite{borgers2015introduction,HuRe06, KaTe15,khalili2019incentivizing, HeAn18a, HeAn20, SiAn20,huang2006auction,HeAn18,rasouli2014electricity,9146379,zhang2020optimal,8824110}, there are a number of agents with some  private information. There is a designer that designs messages together with allocation and tax/subsidy functions.  Agents, as ``senders'' of the information, send  their messages, which could convey true or false information, to the central authority, acting as  the ``receiver'' of the information. Upon reception of these messages, the central authority will then determine their allocations and taxes/subsidies. The incentives for the agents to send truthful messages are created through allocation and tax functions. Note that the central authority commits to the allocation and tax function and can not change these functions after hearing the agents' messages.
	
	In information design~\cite{kamenica2011bayesian,bergemann2019information,farhadi2020dynamic,sayin2021bayesian,anunrojwong2020information,ARXIV-VERSION_HeSa21}, there is usually one ``sender''  who owns a piece of information. The sender shares some part of his information with a number of agents as ``receivers''  by sending messages to them. The messages are created according to a policy that is to be designed by the sender.  The agents will then interpret the messages using the policy based on which the messages are generated and then they take some actions. The sender has to optimally choose his policy to steer agents' actions to a desired direction. Note that similar to the mechanism design framework, the sender commits to the policy he is using to create messages. The difference is that, the commitment in information design is from the sender while in mechanism design it is from the receiver.
	Information design problems with one sender and one receiver are usually referred to as ``Bayesian persuasion''~\cite{kamenica2011bayesian}. In \cite{kamenica2011bayesian}, the authors present a geometric form of interpreting information design and show when it is profitable for the designer to not share some part of the information. However, when there are multiple receivers, the information design problem becomes more complex and notions of equilibria must be introduced to analyze the game. As it is shown in \cite{bergemann2016bayes}, an information design model with multiple receivers is in fact a game with incomplete information and the set of outcomes of the information design problem corresponds to the set of Bayes-correlated equilibria (BCE). In the definition of BCE in \cite{bergemann2016bayes}, the information designer follows a direct strategy where he directly recommends actions to the players.  The strategy of the information designer has to satisfy an obedience condition, that is, each player has to be willing to follow the recommendation.

	In this paper, we combine the two approaches and study joint information and mechanism design for a queuing system. In our model, there is a queue with an unobservable backlog by the incoming users. An information designer observes the queue backlog and makes recommendations to the users arriving at the queue as to whether to join or not to join the queue. This part of the model is an information design setting with the information designer being the sender and the arriving users being the receivers of the information. In our model, the arriving users have payoff relevant private types (we consider a binary type, so users can either be of type~1 or type~2). Upon arrival to the queue, the users send a message to the information designer, that is supposed to be their type. The information designer then creates a recommendation for that specific user type. The users have to pay a tax in exchange for the information they receive. The information sent by the information designer and the tax function incentivize the users to reveal their true types. Note that in this setting, the information designer has two types of commitments. The first commitment is the policy that he uses to generate the recommendations given the queue backlog and declared user type, and the second commitment is the tax function. In this setting, the information designer is a sender of the information in the information design aspect and a receiver in the mechanism design aspect of our model.

	There are some works in information design where similar to our model, the receivers have private information, e.g., private types.  There are two approaches to these problems: without elicitation and with elicitation of the private information. In the case of information design without elicitation~\cite{kolotilin2017persuasion}, the information designer has to send a list of suggestions for each possible type of the receivers. This setting is referred to as public persuasion by Kolotilin et al. in~\cite{kolotilin2017persuasion}.  In the case of information design with elicitation~\cite{kolotilin2017persuasion,bergemann2018design, daskalakis2016does}, receivers report their types and instead of the obedience constraint, the decision rule of the information designer should satisfy an incentive constraint. The incentive constraint makes sure each type of the receiver prefers her own recommendation over other recommendations that she can possibly hear if she reports her type untruthfully. Kolotilin et al.~\cite{kolotilin2017persuasion} refer to this setting as private persuasion.
	In~\cite{bergemann2018design}, authors utilize monetary transfers, i.e., taxes/subsidies, to elicit the private types, as opposed to the model in~\cite{kolotilin2017persuasion} where elicitation is done without taxes.   In \cite{horner2016selling}, the persuasion is done not only through information design0 but also by using monetary transfers. However, the receiver does not have any private information. Similarly, monetary transfers have been utilized for  Bayesian persuasion in~\cite{li2017model} but there is a single receiver and she does not have a private type.
	In~\cite{yamashita2018optimal}, there is also some type of joint mechanism and information design but the information disclosure is public and not a function of the users' reported types. In addition, there are no monetary transfers. In \cite{cai2020third}, authors have studied the effect of a third-party data provider on simple mechanisms and in this sense, they have considered a joint information and mechanism design problem. They show that simple mechanisms fail to approximate the optimal revenue in the presence of a third-party signal.

	Our formulation of joint information and mechanism design is similar in spirit to the one discussed in~\cite{bergemann2018design}, where there are multiple players with private prior beliefs about a state of the world. The information designer offers a menu of experiments (that convey information) that players can choose from and they have to pay a tax in return. The information designer maximizes his revenue over the set of experiments and taxes subject to incentive compatibility and individual rationality constraints. Our setting can be considered a special case of the general framework discussed in~\cite{bergemann2018design}. However, the specifics of our model  such as users' utilities being linear in their private types, enable us to evaluate explicit tax functions and formulate a linear optimization problem for the information designer and study the properties of its solution.
	The queuing system presented in this paper builds on the model by~\cite{lingenbrink2019optimal} with the main difference being that in our model the users have private types where in~\cite{lingenbrink2019optimal}, the incoming traffic is uniform (there are some discussions on the case of different user types but these types are assumed to be known to the information designer).
	
	The rest of the paper is structured as follows. In section~\ref{model}, we discuss the model. In section~\ref{mechobj}, the mechanism objectives are discussed. The tax function is presented in section~\ref{taxfun}. The information design problem is formalized in section~\ref{infodes}. Some structural properties are presented in section~\ref{prop}.  We present results of numerical analysis  in section~\ref{numerical} and we conclude in section~\ref{conc}. The proofs of lemmas and theorems can be found in \optv{submission}{the extended version of this paper \cite{ARXIV-VERSION_HeAn21}.}\optv{arxiv}{the Appendix at the end of the paper.}
	
	%%%%%%%%%%%%%%%%%%%%%%%%%%%%%%%%%%%%%%%%%%%%%%%%%%%%%%%%%%%%%%%%%%%%%%%
	%%%%%%%%%%%%%%%%%%%%%%%%%%%%%%%%%%%%%%%%%%%%%%%%%%%%%%%%%%%%%%%%%%%%%%%

	%%%%%%%%%%%%%%%%%%%%%%%%%%%%%%%%%%%%%%%%%%%%%%%%%%%%%%%%%%%%%%%%%%%%%%%
	%%%%%%%%%%%%%%%%%%%%%%%%%%%%%%%%%%%%%%%%%%%%%%%%%%%%%%%%%%%%%%%%%%%%%%%
	%%%%%%%%%%%%%%%%%%%%%%%%%%%%%%%%%%%%%%%%%%%%%%%%%%%%%%%%%%%%%%%%%%%%%%%
	%%%%%%%%%%%%%%%%%%%%%%%%%%%%%%%%%%%%%%%%%%%%%%%%%%%%%%%%%%%%%%%%%%%%%%%
	\section{Model}\label{model}
	
	We consider a service provider with a service rate of 1.  There is a queue with infinite capacity and users arrive at the queue according to a Poisson arrival process with a rate $\lambda>1$. We denote the number of users in the queue by $x$ and we have $X \sim \mu(\cdot)$, where $ \mu(\cdot)$ is the stationary distribution of the queue backlog. The users have payoff relevant private types $i \in \cI= \{1, 2\}$ and the a-priori type distribution is known $I \sim P_I(\cdot)$.
	%The users have to pay a price $p$ for the service they receive.
	%The utility of a user with type $i$ joining the queue with size $x$ is $i v(x)-p$, where $v(\cdot)$ is a decreasing function.(i \in \cI)
	The queue backlog is unobservable by users. There is an information designer who observes the queue backlog and gives the users an option to hear a recommendation about joining or leaving the queue.  The users have to follow the recommendation if they decide to hear it and they will have to pay a tax in exchange for the recommendation.
	%We denote this mechanism by $\mathbf{M}$ which will be formally defined later on.
	If a user chooses not to hear the recommendation, she will decide whether or not to join the queue on her own. Note that if a user decides not to hear the signal, she joins the same queue that she would have joined had she decided to hear the recommendation.
	%=(m \in \mathcal{I},i \in \mathcal{I},t(\cdot))$.
	One can define the following sequence of actions that users take at the time of arrival. Note that we refer to a user by she and to the information designer by he.

	%\begin{itemize}[leftmargin=*]
\begin{itemize}
		\item	First, the user with type $i$ who has arrived at the queue decides to hear  the recommendation or not by choosing $d=g(i), \ d \in \{0,1\}$, where $d=1$ means she  hears the recommendation and then follows it.

		\item	If the user decides not to hear the recommendation, i.e., $d=0$, she will decide to join or not join the queue by choosing $e=k(i),\  e  \in \{0,1\}$, where $e=1$ means she joins the queue.
		
		\item If the user decides to hear the recommendation, i.e., $d=1$, she has to send a message $m=f(i), \ m \in \cI$ to the information designer.
		%Note that the mechanism $\mathbf{M}$
		This is a direct mechanism and the message $m$ sent by a user with type $i$ is supposed to be her type.
		The information designer determines a tax $t(m)$ that is to be paid by the user in return for hearing the recommendation.
		He then generates a (randomized) recommendation $s$ where $S \sim \sigma(\cdot|x,m)$ and announces it to the user. The distribution  $\sigma(\cdot|x,m)$ is called the recommendation policy. The user will then follow the recommendation, i.e., if $s=1$, she joins the queue  and if $s=0$ she leaves.
	\end{itemize}
	
	Figure~\ref{gametree} depicts the extended form of the game faced by a user with type $i$ at the time of arrival.
	Based on the steps described above, we can denote the private history (or more appropriately, information set) of a user by $h\in \mathcal{H}$, where
	the set of private histories is defined as $\mathcal{H}=\{(i,d=0,e=0)_{(i \in \cI)}, \ (i,d=0,e=1)_{(i \in \cI)}, \ (i,d=1,m,s=0)_{(i,m \in \cI)}, \ (i,d=1,m,s=1)_{(i,m \in \cI)}\}$. The utility of a user for  each of these histories is denoted by $u(h)$, and is described in the following.
	
	The users have to pay a price $p$ for the service they receive if they opt out of the mechanism, but choose to enter the queue. As mentioned before, the users pay a tax $t(m)$ in exchange for the recommendation if they decide to hear it. They do not pay anything extra when they receive service after hearing the recommendation (the price of the service is included in the tax function).
	The user with type $i$ also receives a reward  of $i v(x)$ by joining the queue of backlog $x$, where $v(\cdot)$ is a decreasing function, which can also be negative for large enough $x$.
	Therefore, for $h=(i,d=0,e=1)$, the user receives the expected utility of $u(i,d=0,e=1)=i \bar{v}-p$  where $\bar{v}=\mathbb{E}[v(X)]=\sum_{x=0}^{\infty}v(x)\mu(x)$.
	For $h=(i,d=0,e=0)$, she leaves the queue and receives $u(i,d=0,e=0)=0$.
	For $h=(i,d=1,m,s=1)$, the user receives the expected utility of $u(i,d=1,m,s=1)=i\mathbb{E}[v(X)|h]-t(m)$. Finally, for $h=(i,d=1,m,s=0)$, the user receives the utility $u(i,d=1,m,s=0)=-t(m)$.

	\vspace{-0.3cm}
	\begin{figure}[ht]
		\centering
		\includegraphics[width=8cm]{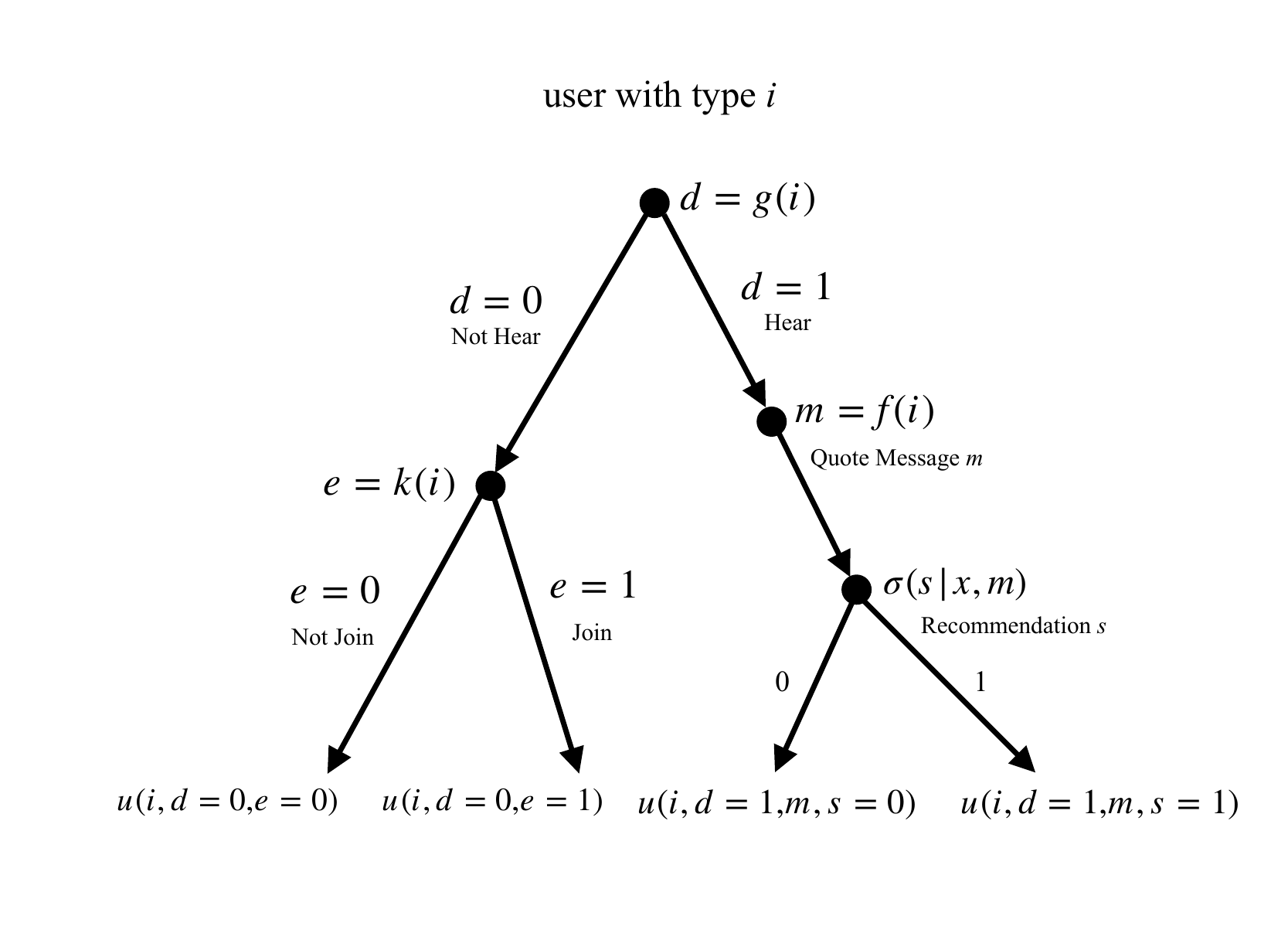}
		\vspace{-0.5cm}
		\caption{Extended form of the game faced by each user at the arrival time. }
		\label{gametree}
	\end{figure}
	
	One can express the joint probability distribution of the random variables described in this model as follows.
	\begin{align}
		\P&(x,i,d,e,m,s)\nonumber \\&=\mu(x)P_I(i)\P(d,e,m|i)\sigma(s|x,m).
		\label{jointdist}
	\end{align}
		Note that $\P(d,e,m|i)$ is determined by the strategy of the user with type $i$ and we have
	\begin{align}
	&\P(m|i)=\mathbf{1}_{f(i)}(m)\\
	&\P(d|i)=\mathbf{1}_{g(i)}(d)\\
	&\P(e|i)=\mathbf{1}_{k(i)}(e),
	\end{align}
	where $\mathbf{1}_{a}(b)=\left\{\begin{array}{cc}
	1 & \quad \text{if } a=b\\
	0 & \quad \text{o.w.}
	\end{array}\right.$
	and the stationary distribution of $X$, $\mu(\cdot)$, depends on $\sigma$ and is characterized in the next section in Lemma~\ref{stationary}.

	\section{Mechanism Objectives}\label{mechobj}
	In the previous section, we described the model and introduced the actions and messages involved in it. To summarize, we have actions/decisions $d=g(i)$, $e=k(i)$, $m=f(i)$ that are taken by the user with type $i$ and $S\sim \sigma(\cdot|x,m)$ and $t(m)$ that are to be designed by the information designer.
	%\red{[AA: Add an initial sentence listing all decision kernels that the agents are deciding and all the ones that the designer is deciding.]}
	
	The mechanism $\mathbf{M}\mdef (\sigma,t)$  is designed by the information designer to have the following properties:
	%\begin{itemize}[leftmargin=*]
\begin{itemize}
		\item IR: The mechanism should be individually rational (IR). That is, both types of users should prefer to hear the recommendation, i.e., $d=g(i)=1, \ \forall i \in \cI$. For the mechanism to be IR, we must have the following condition.
		\begin{align}
			\E(u(i,d=1,M,S) ) \geq \E(u(i,d=0,E) ),  \forall i \in \cI.
		\end{align}
		The above equation implies that the expectation is taken at the step of the game where the user has to decide on $d$ and the decision should always be $d=g(i)=1$.
		
		\item DSIC: The mechanism should be dominant strategy incentive compatible (DSIC). That is, all users that choose to hear the recommendation, should act truthfully no matter what other users do, i.e., $m=f(i)=i$ for all $i$ for which $d=g(i)=1$. For the mechanism to be DSIC we should have the following.
		\begin{align}
			f(i)=\arg \max_m \E(u(i,d=1,m&,S) )=i, \nonumber \\
			&  \forall i \ s.t.\  g(i)=1.
			%f(i)&=\arg \max_m I\sum_{x=0}^{\infty}V(x) \mu(x) \sigma(S=1|x,m) - t(m), \nonumber \\&=I,
		\end{align}
		Note that the utility $u(i,d=1,m,S)$ does not depend on the messages of other users and it only depends on the stationary distribution of $X$, which is affected by the strategies $f(\cdot)$ of users and not the actual messages quoted. Also, note that for whatever strategy $f(\cdot)$ of users, i.e., truthful or not,  it must be  a dominant strategy for each user to quote her message truthfully.

		\item The mechanism should maximize the information designer's expected revenue. That is, the information designer solves the following optimization problem.
		\begin{align}
			\sigma^*,t^*\in \arg\max_{\sigma,t}\lambda\mathbb{E}[ t(M)g(I)],
		\end{align}
		where $M=f(I)$. The above objective comes from the fact that the tax is only paid by the users who participate in the mechanism, i.e., $g(i)=1$.
	\end{itemize}
	
	Given the joint distribution of the random variables in \eqref{jointdist}, we can calculate the expected utilities in IR and DSIC constraints as follows.
	\begin{subequations}
		\begin{align}
			\E&(u(i,d=1,M,S) )\nonumber \\
			=&u(i,d=1,f(i),s=1)\P(S=1|i,d=1,f(i)) \nonumber \\
			&+u(i,d=1,f(i),s=0)\P(S=0|i,d=1,f(i))\\
			=&\sum_{x=0}^{\infty}iv(x)\P(x,S=1|i,d=1,f(i))-t(f(i))\\
			=&\sum_{x=0}^{\infty}iv(x)\mu(x)\sigma(1|x,f(i))-t(f(i))
		\end{align}
	\end{subequations}
	\begin{subequations}
		\begin{align}
			\E&(u(i,d=0,E) )
			=u(i,d=0,e=1)\mathbf{1}_{k(i)}(e=1)\\
			&=k(i) \sum_{x=0}^{\infty}(iv(x)-p)\mu(x)=k(i) (i\bar{v}-p).
			\label{eq:util_e1}
		\end{align}
	\end{subequations}
	
	The following lemma characterizes the function $k(i)$.
	\begin{lemma}
		The function $k(i)$ is given by the following equation.
		\begin{align}
			k(i)=\left\{\begin{array}{cc}
				1 & \quad \text{if } \ i\bar{v}-p\geq 0\\
				0 & \quad \text{o.w.}
			\end{array}\right.
		\end{align}
	\end{lemma}

	\begin{IEEEproof}
		The result is evident by comparing the utility in~\eqref{eq:util_e1} with 0.
	\end{IEEEproof}

	One can use this lemma and the expected values of the utilities to further simplify the IR and DSIC condition as follows.
	
	\begin{subequations}
		IR:
		\begin{align}
			\sum_{x=0}^{\infty}& i v(x)\mu(x)\sigma(1|x,i)-t(i) \nonumber  \\
			&\geq
			(\sum_{x=0}^{\infty}iv(x)\mu(x)-p) ^+\quad \forall i \in \cI.
			\label{IR}
		\end{align}
		
		DSIC:
		\begin{align}
			f(i) &=\arg \max_m\sum_{x=0}^{\infty} i v(x) \mu(x)\sigma(1|x,m)-t(m) \nonumber \\
			&=i,  \qquad \forall i \in \cI,
		\end{align}
	\end{subequations}
	where $(a)^+=\max(a,0)$.
	%\red{[AA: can we switch the order and then in the IR part substitute $f(i)=i$. Do we need this? Is this going to make the remaining development easier?]}
	
	Given IR and DSIC constraints, the optimization problem of the information designer can be simplified as follows.
	%\begin{subequations}
	\begin{align}
		\sigma^*,t^*\in \max_{\sigma,t}\lambda \mathbb{E}[t(I)]
		=\max_{\sigma,t}\lambda (P_I(1)t(1)+P_I(2)t(2)).
	\end{align}
	%\end{subequations}

	In the expected utilities, we see the stationary distribution of $X$,  $\mu(\cdot)$ which is characterized in the following lemma.
	\begin{lemma}
		If we assume that IR and DSIC hold, the stationary distribution of $X$, $\mu(\cdot)$, is given by the following equation.
		\begin{subequations}
			\begin{align}
				&\mu(x+1)\hspace{-0.05cm} =\hspace{-0.05cm} \lambda \mu(x)(P_I(1)\sigma(1|x,1)+P_I(2)\sigma(1|x,2))\\
				&\sum_{x=0}^{\infty}\mu(x)\hspace{-0.05cm} =1.
			\end{align}
		\end{subequations}
		\label{stationary}
	\end{lemma}
	
%	
%	\begin{IEEEproof}
%		See Appendix.
%	\end{IEEEproof}
%	

	\section{Tax Function}\label{taxfun}
	In this section we introduce the tax functions for the mechanism and prove DSIC. We will also prove that this type of tax function maximizes the designer's revenue.
	
	We consider the following tax function.
	\begin{subequations}
		\label{tax}
		\begin{align}
			t(1)=&t_0+q(1)\\
			t(2)=&t_0+2 q(2)-q(1), %\\
			%    =&t(0)+m\sum_{x=0}^{\infty}v(x) \mu(x) \sigma(1|x,m)\nonumber \\
			%         &-\sum_{j=0}^{m-1}\sum_{x=0}^{\infty}v(x) \mu(x) \sigma(1|x,j),
		\end{align}
	\end{subequations}
	where we define
	\begin{align}
		q(i)=\sum_{x=0}^{\infty}v(x) \mu(x) \sigma(1|x,i), \ \forall i \in \cI,
	\end{align}
	which can be viewed as the ``allocation'' to a user with quoted message $i$. 	We refer to $t_0$ as the tax offset.
	
	Notice that there are two degrees of freedom in the tax functions, $t_0$ and $q$, which is determined by $\sigma$. We will see in the next theorem that the given  tax function guarantees DSIC. The other two requirements of the mechanism, i.e., IR and revenue maximization will determine the two degrees of freedom in the tax function.
	
	%	\red{AA: BTW, this design of taxes result in DSIC ie, dominant strategy IC, which is quite strong! correct? maybe we need to emphasize this earlier in the design objectives and use DSIC everywhere instead of IC.\\
	%		NH: In Mayerson'e Lemma, I believe we can say it results in DSIC. However, in our allocation rule, we consider $\mu$ to be a known (fixed) quantity, while it is derived with the assumption of users being truthful.  On the other hand, the result holds for any $\mu$. So let me think about it more carefully.\\
	%		AA: The DSIC refers only to the types reported by other users. I guess here we wrote the allocation function as $q(m)$, ie, that depends only on the message of the current user, but is there dependence on other messages as well?
	%		\\
	%		NH: There is no dependence of $q(m)$ on the messages of others. However, $q(m)$ depends on the strategies with wich others' messages are generated. This is because the strategies affect the stationary distribution and therefore, they affect $q(m)$. But, we can say that for any stationary distribution (i.e., for every strategy of others), it is a dominent strategy to quote truthful. Therefore, we can claim DSIC. }
	%	
	
	\begin{theorem}
		The mechanism $\mathbf{M}$ is DSIC if
		\begin{itemize}
			\item $q(2)\geq q(1)$.
			\item the tax function is given by  equation \eqref{tax}.
		\end{itemize}
		Furthermore, the given tax function  maximizes the information designer's revenue.
		\label{IC}
	\end{theorem}
%	\begin{IEEEproof}
%		See Appendix.
%	\end{IEEEproof}
	%\red{Can we put this statement together with Theorem 1 and have a more concrete claim that (11a) is the tax function that guarantees DSIC and maximization of revenue...}	

	\section{Information Design  Problem}\label{infodes}
	Given the tax function described in the previous section and the fact that DSIC holds for  a mechanism with such tax function, and if we assume uniform distribution for the types of users, i.e., $P_I(1)=P_I(2)=\frac{1}{2}$, we can simplify the objective of the information designer as follows.
	\begin{subequations}
		\begin{align}
			\lambda \mathbb{E}[t(I)]=&\frac{\lambda}{2} ( t_0+q(1)+ t_0+2q(2)-q(1))\\
			=&\lambda (t_0+q(2)).
		\end{align}
	\end{subequations}
	Therefore, by including the constraints,  we have the following optimization problem.
	\begin{subequations}
		\begin{align}
			\max&_{\sigma,t_0} \lambda( t_0+q(2))\\
			s.t. & \quad  -t_0 \geq  (\bar{v}-p)^+ \label{ircona}\\
			& \quad  q(1)-t_0 \geq  (2\bar{v}-p)^+ \label{irconb} \\
			& \quad q(2) \geq q(1) \\
			& \quad  \mu(x+1)=\lambda \mu(x)\frac{\sigma(1|x,1)+\sigma(1|x,2)}{2} \label{stationarycon}\\
			& \quad  \sum_{x=0}^{\infty}\mu(x)=1.
		\end{align}
		\label{opt-nonlin}
	\end{subequations}
	Note that constraints \eqref{ircona} and \eqref{irconb} enforce the IR condition.
	The above optimization problem is not linear with respect to $\sigma$ because of the stationary constraints of \eqref{stationarycon}. Furthermore, the constraints are not expressed as linear inequalities. Therefore, we restate the problem in terms of the joint probability distribution on $(S,I,X)$, denoted by $\gamma(S,I,X)$ to have a linear optimization problem. Note that $q(i)=\sum_{x=0}^{\infty}v(x) \mu(x) \sigma(1|x,i)=\sum_{x=0}^{\infty}v(x) \frac{\gamma(s=1,i,x)}{P_I(i)}$.  Therefore, we have the following linear optimization problem faced by the information designer.
	\begin{subequations}
		\begin{align}
			\hspace{-0.5cm}\max&_{\gamma,t_0} \lambda t_0+2\lambda \sum_{x=0}^{\infty}v(x) \gamma(s=1,2,x)\\
			\text{s.t.}   \ &   -t_0  \geq  \sum_{x=0}^{\infty}v(x) \sum_{s,i}\gamma(s,i,x)-p \quad \label{linear_b}\\
			& \  2\sum_{x=0}^{\infty}v(x) \gamma(s=1,1,x)-t_0 \nonumber \\ & \hspace{2.4cm} \geq  2\sum_{x=0}^{\infty}v(x) \sum_{s,i}\gamma(s,i,x)-p \label{linear_c}\\
			& \ -t_0 \geq 0 \label{linear_d}\\
			& \   2\sum_{x=0}^{\infty}v(x) \gamma(s=1,1,x)-t_0 \geq 0  \label{linear_e}\\
			& \  \sum_{x=0}^{\infty}v(x) \gamma(s=1,2,x) \hspace{-0.05cm}\geq \sum_{x=0}^{\infty}v(x) \gamma(s=1,1,x)  \\
			& \   \sum_{s,i}\gamma(s,i,x+1)=\lambda \sum_{i}\gamma(1,i,x), \ \forall x\geq 0\\
			& \ \sum_s\gamma(s,i,x)=\frac{1}{2}\sum_{s,i} \gamma(s,i,x),\ \forall i \in \cI,  x \geq 0 \label{indxi}\\
			& \   \sum_{s,i,x}\gamma(s,i,x)=1\\
			& \   \gamma(s,i,x)\geq 0, \quad  \forall s \in \{0,1\}, \ i \in \cI, \ x \geq 0.
			\label{objconst}
		\end{align}
		\label{opt-lin}
	\end{subequations}
	Note that constraints~\eqref{linear_b}, \eqref{linear_c}, \eqref{linear_d}, \eqref{linear_e} correspond to linearized constraints of the IR condition, while constraint~\eqref{indxi} is to ensure $\P(x,i)=\mu(x)P_I(i)$ according to~\eqref{jointdist}.

	\section{Structural Properties}\label{prop}
	In this section, we discuss some properties and behaviors of the optimal recommendation policy that is the solution of the optimization problem \eqref{opt-lin}.
	%\textcolor{red}{[AA: this theorem is derived using KKT conditions for the problem in (15). What do we know about uniqueness and existence of solutions?]}
	\begin{theorem}
		Suppose $t_0$ and $\gamma^*$ (or equivalently $\sigma^*$) are the solution of \eqref{opt-lin}. Then, one of the following holds.
		%\begin{itemize}[leftmargin=*]
\begin{itemize}
			\item Case 1:
			\begin{itemize}
				\item If $v(x)> 0$ and $\sigma(s=1|1,x)>0$, then $\sigma(s=1|2,x)=1$.
				\item If $v(x)<0$ and $\sigma(s=1|2,x)>0$, then $\sigma(s=1|1,x)=1$.
			\end{itemize}
			\item Case 2:	There is a threshold $\tilde{x}$ such that  for $x\geq \tilde{x}$ we have $\sigma(1|i,x)=0$  for all $i \in \cI$. Furthermore,  for $x<\tilde{x}$,   $\sigma(1|i,x)=1$ for all $i \in \cI$ except for some points in $\tilde{\cX}=\{x_1, x_2, \ldots\}, \ x_k <\tilde{x} $ for which we can have $\sigma(1|1,x_k)<1$, or $\sigma(1|2,x_k)<1$, where  all $x_k \in \tilde{\cX}$ satisfy the following condition. There exists $\epsilon_1>0$ and $\psi$ such that
			\begin{align}
				&(2\sum_{x=0}^{x_k} \lambda^x v(x))\epsilon_1 + (\sum_{x=0}^{x_k} \lambda^x )\psi=\sum_{x=0}^{x_k-1} \lambda^x v(x),   \forall x_k \in \tilde{\cX}. \label{epsilon_psi-p}
			\end{align}
		\end{itemize}
		\label{structure}
	\end{theorem}
%	\begin{IEEEproof}
%		See Appendix.
%	\end{IEEEproof}
	
	The intuitive explanation of the above theorem is as follows.
	If we define $x_0$ as $v(x)>0$ for $x< x_0$ and $v(x)<0$ for $x > x_0$, then case 1 of the theorem implies a threshold behavior for the recommendation policy with the threshold being $x_0$. That is, for $x$ below the threshold, the recommendation policy favors type 2 of the users and only allows type 1 to enter the queue if type 2 is allowed in with probability 1.  Similarly, for $x$ above the threshold, type 1 is favored for entering the queue and type 2 is allowed to enter the queue if type 1 is allowed in with probability 1.
	
	Case 2 of the theorem implies that the designer is sending the same signal (except for $x\in \tilde{\cX}$) for both types of  users. Therefore, revenue due to the discrimination between the two user types can only be gained for states  $x\in \tilde{\cX}$. One question that arises is what the size $|\tilde{\mathcal{X}}|$ of this set  is.
	Equation~\eqref{epsilon_psi-p} indicates that for a given size $|\tilde{\mathcal{X}}|$,  there are $|\tilde{\mathcal{X}}|$ equations to be satisfied and only two unknowns ($\epsilon_1$ and $\psi$). As a result, it is highly unlikely that for a general utility function $v(\cdot)$, the size of the set is larger than 2.
	Evaluating the quantities $x_1$ and $x_2$ can be done systematically by first evaluating $\tilde{x}$ and then searching over all
	$O(\tilde{x}^2)$ cases for the values of $x_1$ and $x_2$ by checking if~\eqref{epsilon_psi-p} is satisfied for some $\epsilon_1>0$ and $\psi$.

	\section{Numerical Analysis}\label{numerical}
	In this section, we present some numerical analysis of the model discussed in this paper. We
	have numerically solved the linear optimization problem \eqref{opt-lin} using Matlab. In our analysis, we have set a maximum capacity for the queue such that it does not affect the stationary distribution of the queue backlog.
	We consider $\lambda=1.2$, and $v(x)=1-(x/50)^2$ as the utility function of the users. Fig.~\ref{signaling0} depicts the recommendation policy of the information designer with respect to the queue backlog for  $p=0$, while  the stationary distribution of the queue backlog is plotted in the same figure. In this case, the revenue of the designer  is $0.0786$. The recommendation policy in Fig.~\ref{signaling0} confirms the results of Theorem~\ref{structure}. The threshold $x_0$ in this example is $x_0=50$ and we can see for $x<x_0$, both types are allowed in with probability 1, and for $x > x_0$, either both types are allowed in with probability 1 (states 51-53) or only type 1 is allowed in. This is consistent with case~1 of Theorem~\ref{structure}.
	In Fig.~\ref{signaling02} we plot similar quantities for the case of $p=0.2$ and the revenue of the information designer is $0.2693$. Similar to the $p=0$ case, the results are consistent with case~1 of Theorem~\ref{structure} because for $x<x_0$, type 2 is allowed in with probability 1, and for $x>x_0$ type 1 is allowed in with probability 1.
	
	In order to evaluate how well the information designer is doing in terms of gaining revenue,  we can calculate the revenue of the queue when there is no information designer and the incoming traffic chooses to join the queue without any information. For  the given $v(\cdot)$, only one type of users will join the queue (type~2) and the rate would become $\frac{\lambda}{2}$. Otherwise, the queue becomes unstable. Therefore, the revenue  is at most $\frac{ \lambda p}{2}$. Hence,  for $p=0$, the revenue of the outside option is $0$ and for $p=0.2$, the revenue of the outside option is $0.12$ and clearly, the information designer is doing better than the outside option.

	\begin{figure}
		\centering
		\begin{tikzpicture}
		% defining custom colors
		\definecolor{mycolor1}{rgb}{0.15,0.15,0.15}
		\definecolor{mycolor2}{rgb}{0,0.447,0.741}
		\definecolor{mycolor3}{rgb}{0.85,0.325,0.098}
		\definecolor{mycolor4}{rgb}{0.929,0.694,0.125}

		% Axis at [0.13 0.11 0.78 0.82]
		\begin{axis}[
		scale only axis,
		every outer x axis line/.append style={mycolor1},
		every x tick label/.append style={font=\color{mycolor1}},
		every outer y axis line/.append style={mycolor1},
		xlabel={$x$},
		every y tick label/.append style={font=\color{mycolor1}},
		width=2.7in,
		height=1.6in,
		xmin=0, xmax=120,
		ymin=0, ymax=1.2,
		axis on top]
		
		\addplot [
		color=mycolor2,
		solid,
		mark=asterisk,
		mark options={solid}
		]
		coordinates{
			(1,1)
			(2,1)
			(3,1)
			(4,1)
			(5,1)
			(6,1)
			(7,1)
			(8,1)
			(9,1)
			(10,1)
			(11,1)
			(12,1)
			(13,1)
			(14,1)
			(15,1)
			(16,1)
			(17,1)
			(18,1)
			(19,1)
			(20,1)
			(21,1)
			(22,1)
			(23,1)
			(24,1)
			(25,1)
			(26,1)
			(27,1)
			(28,1)
			(29,1)
			(30,1)
			(31,1)
			(32,1)
			(33,1)
			(34,1)
			(35,1)
			(36,1)
			(37,1)
			(38,1)
			(39,1)
			(40,1)
			(41,1)
			(42,1)
			(43,1)
			(44,1)
			(45,1)
			(46,1)
			(47,1)
			(48,1)
			(49,1)
			(50,1)
			(51,1)
			(52,1)
			(53,1)
			(54,1)
			(55,1)
			(56,1)
			(57,1)
			(58,1)
			(59,1)
			(60,1)
			(61,1)
			(62,1)
			(63,1)
			(64,1)
			(65,1)
			(66,1)
			(67,1)
			(68,1)
			(69,1)
			(70,1)
			(71,1)
			(72,1)
			(73,1)
			(74,1)
			(75,1)
			(76,1)
			(77,0)
			
		};
		
		\addplot [
		color=mycolor3,
		solid,
		mark=+,
		mark options={solid}
		]
		coordinates{
			(1,1)
			(2,1)
			(3,1)
			(4,1)
			(5,1)
			(6,1)
			(7,1)
			(8,1)
			(9,1)
			(10,1)
			(11,1)
			(12,1)
			(13,1)
			(14,1)
			(15,1)
			(16,1)
			(17,1)
			(18,1)
			(19,1)
			(20,1)
			(21,1)
			(22,1)
			(23,1)
			(24,1)
			(25,1)
			(26,1)
			(27,1)
			(28,1)
			(29,1)
			(30,1)
			(31,1)
			(32,1)
			(33,1)
			(34,1)
			(35,1)
			(36,1)
			(37,1)
			(38,1)
			(39,1)
			(40,1)
			(41,1)
			(42,1)
			(43,1)
			(44,1)
			(45,1)
			(46,1)
			(47,1)
			(48,1)
			(49,1)
			(50,1)
			(51,1)
			(52,1)
			(53,1)
			(54,0.136193)
			(55,0)
			
		};
		
		\addplot [
		color=mycolor4,
		solid,
		mark=x,
		mark options={solid}
		]
		coordinates{
			(1,8.25431e-06)
			(2,9.90517e-06)
			(3,1.18862e-05)
			(4,1.42634e-05)
			(5,1.71161e-05)
			(6,2.05394e-05)
			(7,2.46472e-05)
			(8,2.95767e-05)
			(9,3.5492e-05)
			(10,4.25904e-05)
			(11,5.11085e-05)
			(12,6.13302e-05)
			(13,7.35962e-05)
			(14,8.83155e-05)
			(15,0.000105979)
			(16,0.000127174)
			(17,0.000152609)
			(18,0.000183131)
			(19,0.000219757)
			(20,0.000263709)
			(21,0.00031645)
			(22,0.00037974)
			(23,0.000455688)
			(24,0.000546826)
			(25,0.000656191)
			(26,0.00078743)
			(27,0.000944916)
			(28,0.0011339)
			(29,0.00136068)
			(30,0.00163281)
			(31,0.00195938)
			(32,0.00235125)
			(33,0.0028215)
			(34,0.0033858)
			(35,0.00406296)
			(36,0.00487556)
			(37,0.00585067)
			(38,0.0070208)
			(39,0.00842496)
			(40,0.01011)
			(41,0.0121319)
			(42,0.0145583)
			(43,0.01747)
			(44,0.020964)
			(45,0.0251568)
			(46,0.0301882)
			(47,0.0362258)
			(48,0.043471)
			(49,0.0521651)
			(50,0.0625982)
			(51,0.0751178)
			(52,0.0901414)
			(53,0.10817)
			(54,0.129804)
			(55,0.0884891)
			(56,0.0530935)
			(57,0.0318561)
			(58,0.0191136)
			(59,0.0114682)
			(60,0.00688091)
			(61,0.00412855)
			(62,0.00247713)
			(63,0.00148628)
			(64,0.000891766)
			(65,0.00053506)
			(66,0.000321036)
			(67,0.000192622)
			(68,0.000115573)
			(69,6.93438e-05)
			(70,4.16063e-05)
			(71,2.49638e-05)
			(72,1.49783e-05)
			(73,8.98695e-06)
			(74,5.39217e-06)
			(75,3.2353e-06)
			(76,1.94118e-06)
			(77,1.76804e-22)
			(78,0)
			
		};
		
		\addlegendimage{/pgfplots/refstyle=Lit}\addlegendentry{\tiny{$\sigma(1|x,1)$}}
		\addlegendimage{/pgfplots/refstyle=Lit}\addlegendentry{\tiny{$\sigma(1|x,2)$}}
		\addlegendimage{/pgfplots/refstyle=Lit}\addlegendentry{\tiny{$\mu(x)$}}
		\end{axis}
		
		\end{tikzpicture}
		\caption{Recommendation policy and the stationary distribution of the queue backlog for $p=0$.}
		\label{signaling0}
	\end{figure}
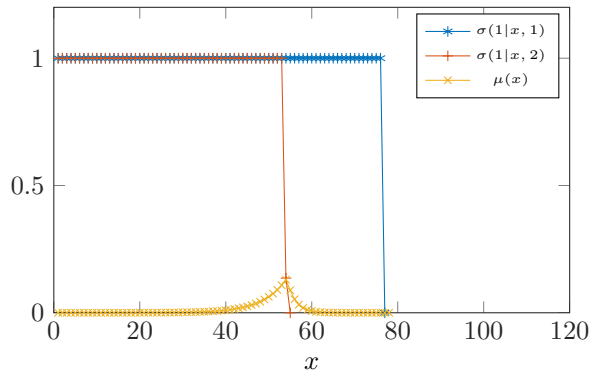

	\begin{figure}
		\centering
		\begin{tikzpicture}
		
		% defining custom colors
		\definecolor{mycolor1}{rgb}{0.15,0.15,0.15}
		\definecolor{mycolor2}{rgb}{0,0.447,0.741}
		\definecolor{mycolor3}{rgb}{0.85,0.325,0.098}
		\definecolor{mycolor4}{rgb}{0.929,0.694,0.125}
		
		\begin{axis}[
		scale only axis,
		every outer x axis line/.append style={mycolor1},
		every x tick label/.append style={font=\color{mycolor1}},
		every outer y axis line/.append style={mycolor1},
		every y tick label/.append style={font=\color{mycolor1}},
		xlabel={$x$},
		width=2.7in,
		height=1.6in,
		xmin=0, xmax=180,
		ymin=0, ymax=1.2,
		axis on top]
		
		\addplot [
		color=mycolor2,
		solid,
		mark=asterisk,
		mark options={solid}
		]
		coordinates{
			(1,1)
			(2,1)
			(3,1)
			(4,1)
			(5,1)
			(6,1)
			(7,1)
			(8,1)
			(9,1)
			(10,1)
			(11,1)
			(12,1)
			(13,1)
			(14,1)
			(15,1)
			(16,1)
			(17,1)
			(18,1)
			(19,1)
			(20,1)
			(21,1)
			(22,1)
			(23,1)
			(24,1)
			(25,1)
			(26,1)
			(27,1)
			(28,1)
			(29,1)
			(30,1)
			(31,1)
			(32,1)
			(33,1)
			(34,1)
			(35,1)
			(36,0)
			(37,0)
			(38,0)
			(39,0)
			(40,0)
			(41,0)
			(42,0)
			(43,0)
			(44,0)
			(45,0)
			(46,0.976527)
			(47,1)
			(48,1)
			(49,1)
			(50,1)
			(51,1)
			(52,1)
			(53,1)
			(54,1)
			(55,1)
			(56,1)
			(57,1)
			(58,1)
			(59,1)
			(60,1)
			(61,1)
			(62,1)
			(63,1)
			(64,1)
			(65,1)
			(66,1)
			(67,1)
			(68,1)
			(69,1)
			(70,1)
			(71,1)
			(72,1)
			(73,1)
			(74,1)
			(75,1)
			(76,1)
			(77,1)
			(78,1)
			(79,1)
			(80,1)
			(81,1)
			(82,1)
			(83,1)
			(84,1)
			(85,1)
			(86,1)
			(87,1)
			(88,1)
			(89,1)
			(90,1)
			(91,1)
			(92,1)
			(93,1)
			(94,1)
			(95,1)
			(96,1)
			(97,1)
			(98,1)
			(99,1)
			(100,1)
			(101,1)
			(102,1)
			(103,1)
			(104,1)
			(105,1)
			(106,1)
			(107,1)
			(108,1)
			(109,1)
			(110,1)
			(111,1)
			(112,1)
			(113,1)
			(114,1)
			(115,1)
			(116,1)
			(117,1)
			(118,1)
			(119,1)
			(120,1)
			(121,0)
		};

		\addplot [
		color=mycolor3,
		solid,
		mark=+,
		mark options={solid}
		]
		coordinates{
			(1,1)
			(2,1)
			(3,1)
			(4,1)
			(5,1)
			(6,1)
			(7,1)
			(8,1)
			(9,1)
			(10,1)
			(11,1)
			(12,1)
			(13,1)
			(14,1)
			(15,1)
			(16,1)
			(17,1)
			(18,1)
			(19,1)
			(20,1)
			(21,1)
			(22,1)
			(23,1)
			(24,1)
			(25,1)
			(26,1)
			(27,1)
			(28,1)
			(29,1)
			(30,1)
			(31,1)
			(32,1)
			(33,1)
			(34,1)
			(35,1)
			(36,1)
			(37,1)
			(38,1)
			(39,1)
			(40,1)
			(41,1)
			(42,1)
			(43,1)
			(44,1)
			(45,1)
			(46,1)
			(47,1)
			(48,1)
			(49,1)
			(50,1)
			(51,1)
			(52,1)
			(53,1)
			(54,1)
			(55,1)
			(56,1)
			(57,1)
			(58,1)
			(59,1)
			(60,1)
			(61,1)
			(62,1)
			(63,1)
			(64,1)
			(65,1)
			(66,1)
			(67,1)
			(68,1)
			(69,1)
			(70,0.85693)
			(71,0)
			
		};
		
		\addplot [
		color=mycolor4,
		solid,
		mark=x,
		mark options={solid}
		]
		coordinates{
			(1,0.000145696)
			(2,0.000174835)
			(3,0.000209802)
			(4,0.000251763)
			(5,0.000302115)
			(6,0.000362538)
			(7,0.000435046)
			(8,0.000522055)
			(9,0.000626466)
			(10,0.000751759)
			(11,0.000902111)
			(12,0.00108253)
			(13,0.00129904)
			(14,0.00155885)
			(15,0.00187062)
			(16,0.00224474)
			(17,0.00269369)
			(18,0.00323243)
			(19,0.00387891)
			(20,0.00465469)
			(21,0.00558563)
			(22,0.00670276)
			(23,0.00804331)
			(24,0.00965197)
			(25,0.0115824)
			(26,0.0138988)
			(27,0.0166786)
			(28,0.0200143)
			(29,0.0240172)
			(30,0.0288206)
			(31,0.0345848)
			(32,0.0415017)
			(33,0.0498021)
			(34,0.0597625)
			(35,0.071715)
			(36,0.086058)
			(37,0.0516348)
			(38,0.0309809)
			(39,0.0185885)
			(40,0.0111531)
			(41,0.00669187)
			(42,0.00401512)
			(43,0.00240907)
			(44,0.00144544)
			(45,0.000867266)
			(46,0.00052036)
			(47,0.000617103)
			(48,0.000740524)
			(49,0.000888628)
			(50,0.00106635)
			(51,0.00127962)
			(52,0.00153555)
			(53,0.00184266)
			(54,0.00221119)
			(55,0.00265343)
			(56,0.00318412)
			(57,0.00382094)
			(58,0.00458513)
			(59,0.00550215)
			(60,0.00660258)
			(61,0.0079231)
			(62,0.00950772)
			(63,0.0114093)
			(64,0.0136911)
			(65,0.0164293)
			(66,0.0197152)
			(67,0.0236583)
			(68,0.0283899)
			(69,0.0340679)
			(70,0.0408815)
			(71,0.0455484)
			(72,0.027329)
			(73,0.0163974)
			(74,0.00983845)
			(75,0.00590307)
			(76,0.00354184)
			(77,0.00212511)
			(78,0.00127506)
			(79,0.000765038)
			(80,0.000459023)
			(81,0.000275414)
			(82,0.000165248)
			(83,9.9149e-05)
			(84,5.94894e-05)
			(85,3.56936e-05)
			(86,2.14162e-05)
			(87,1.28497e-05)
			(88,7.70982e-06)
			(89,4.62589e-06)
			(90,2.77554e-06)
			(91,1.66532e-06)
			(92,9.99193e-07)
			(93,5.99516e-07)
			(94,3.59709e-07)
			(95,2.15826e-07)
			(96,1.29495e-07)
			(97,7.76972e-08)
			(98,4.66183e-08)
			(99,2.7971e-08)
			(100,1.67826e-08)
			(101,1.00696e-08)
			(102,6.04174e-09)
			(103,3.62504e-09)
			(104,2.17503e-09)
			(105,1.30502e-09)
			(106,7.83009e-10)
			(107,4.69806e-10)
			(108,2.81883e-10)
			(109,1.6913e-10)
			(110,1.01478e-10)
			(111,6.08868e-11)
			(112,3.65321e-11)
			(113,2.19192e-11)
			(114,1.31515e-11)
			(115,7.89093e-12)
			(116,4.73456e-12)
			(117,2.84073e-12)
			(118,1.70444e-12)
			(119,1.02266e-12)
			(120,6.13599e-13)
			(121,3.68159e-13)
			(122,0)
			
		};
		
		\addlegendimage{/pgfplots/refstyle=Lit}\addlegendentry{\tiny{$\sigma(1|x,1)$}}
		\addlegendimage{/pgfplots/refstyle=Lit}\addlegendentry{\tiny{$\sigma(1|x,2)$}}
		\addlegendimage{/pgfplots/refstyle=Lit}\addlegendentry{\tiny{$\mu(x)$}}
		\end{axis}
		
		\end{tikzpicture}
		\caption{Recommendation policy and the stationary distribution of the queue backlog for $p=0.2$.}
		\label{signaling02}
	\end{figure}
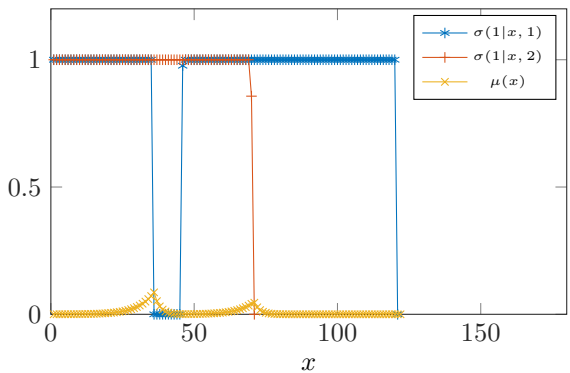

	\section{Conclusion} \label{conc}
	In this paper, we studied an information design problem for a queuing system where in addition to the information designer that privately observes the queue backlog,  the users also have payoff relevant private types. Therefore, a joint information and mechanism design problem was studied. We investigated how the information designer can design tax functions and provide different information for different types of users in order to gain the most revenue. Some structural results were provided for the optimal recommendation policy of the information designer and numerical analysis was done to support the results.

	%\vspace{-0.5cm}
	%	\IEEEtriggeratref{35}
	\optv{arxiv}{
	\appendix
	
	\subsection{Proof of Lemma \ref{stationary}}
		We have a continuous time M/M/1 queue with transition rates $g_{x,x+1}=\lambda \P(S=1|x)$
		%i.e., rate of transiting from state $x$ to state $x+1$
	and $g_{x,x-1}=1$ for $x>0$. Therefore, we have $g_{x,x}=-1-\lambda \P(S=1|x)$ for $x>0$ and $g_{0,0}=-\lambda \P(S=1|0)$. We can calculate the stationary distribution of the queue as follows.
		\begin{subequations}
			\begin{align}
				&\sum_{j}\mu(j)g_{j,x}=0, \quad \forall x\geq 0,
			\end{align}
			which can be written explicitly as
			\begin{align}
				\mu(1) &=\lambda \P(S=1|0) \mu(0)\\
				\mu(x+1)&+\lambda \P(S=1|x-1) \mu(x-1) \nonumber \\
				&=(1+\lambda \P(S=1|x))\mu(x), \quad x\geq 1,
			\end{align}
			and after substituting recursively we find
			\begin{align}
				&\mu(x+1) =\lambda \P(S=1|x) \mu(x)= \\
				&\lambda \mu(x) (P_I(1)\sigma(1|x,1)+P_I(2)\sigma(1|x,2)), \ \forall x\geq 0.
			\end{align}
		\end{subequations}
	%\end{IEEEproof}

	\subsection{Proof of Theorem \ref{IC}}
		This theorem can be proved using Myerson's Lemma \cite{myerson1981optimal}. In order to see the connection, note that given the tax function described in equation \eqref{tax}, one can write the following  for $\E(u(i,d=1,m,S))$.
		\begin{align}
			\E(u(i,d=1,m,S))\hspace{-0.05cm}
		=(i-m)q(m)+	\hspace{-0.05cm}\sum_{j=1}^{m-1}	\hspace{-0.05cm}q(j)\hspace{-0.05cm}-\hspace{-0.05cm}t_0,
			\label{util}
		\end{align}
		where $m=f(i)$ and	if $m=f(i)=i$, we have
		\begin{align}
			\E(u(i,d=1,m=i,S))&=\sum_{j=1}^{i-1}q(j)-t_0.
		\end{align}
		By comparing the above utility, which is gained by a user if she acts truthfully, with the one in \eqref{util} for any $m\neq i$, we can see that if $q(i)$ is increasing in $i$, then DSIC holds.

		In order to prove the second part of the theorem, note that because of the discrete type space, we do not have uniqueness for the tax functions. In other words, revenue equivalence theorem does not hold. However, we can show that the tax function defined in equation \eqref{tax} is an upper bound on all of the tax functions satisfying DSIC and therefore, it is the best the designer can do in terms of maximizing his revenue.
		
		One can write the following for any tax function satisfying DSIC.
		\begin{align}
			t(m+1)-t(m) &\geq m[q(m+1)-q(m)] \\
			t(m+1)-t(m) &\leq (m+1) [q(m+1)-q(m)].
		\end{align}
		Therefore, we have $	t(2) \leq t(1)+ 2 [q(2)-q(1)]$.
		We can see that if tax functions are defined according to equation \eqref{tax}, we have $t(2) =t(1)+ 2 [q(2)-q(1)]$. That is, given $t(1)$, the upper bound is reached for $t(2)$. On the other hand, according to IR constraint, we must have $t(1)\leq q(1)$ and according to  \eqref{tax}, we have $t(1)=t_0+q(1)$, where $t_0$ is optimized over subject to IR constraint. Therefore, $t(1)$ is also its maximum possible value.
		Therefore, the designer can not gain any more revenue using other forms of the tax function.
	%\end{proof}
	
	\subsection{Proof of Theorem \ref{structure}}
		Since optimization problem \eqref{opt-lin} is linear in $\gamma$, we can use KKT conditions to characterize the solution. We use the following dual variables for each constraints in \eqref{opt-lin}.
		%\vspace{-0.4cm}
		\begin{align}
			&   -2\sum_{x=0}^{\infty}v(x) \gamma(s=1,1,x)+ 2\sum_{x=0}^{\infty}v(x) \sum_{s,i}\gamma(s,i,x)\nonumber \\& \hspace{5cm}-p+t_0 \leq 0  :\ \epsilon_1  \\
			&\sum_{x=0}^{\infty}v(x) \sum_{s,i}\gamma(s,i,x)-p+t_0 \leq 0  :\ \epsilon_2 \\
			& -2\sum_{x=0}^{\infty}\hspace{-0.05cm}v(x) \gamma(s=1,1,x)+ t_0 \leq 0  :\  \epsilon_3   \\
			&  t_0\leq 0  : \ \epsilon_4  \\
			&  \sum_{x=0}^{\infty}\hspace{-0.1cm} v(x) \gamma(s=1,1,x) \hspace{-0.05cm} - \hspace{-0.05cm}\sum_{x=0}^{\infty} \hspace{-0.05cm} v(x) \gamma(s=1,2,x) \hspace{-0.05cm} \leq \hspace{-0.05cm} 0 : \eta  \\
			%&\hspace{5cm} \rightarrow \quad \eta\\
			&  \sum_{s,i}\gamma(s,i,x+1)\hspace{-0.05cm} -\hspace{-0.05cm} \lambda \sum_{i}\gamma(1,i,x)=0, \ \forall x\geq 0 : \alpha_x
			%\\& \hspace{5cm} \rightarrow \quad \alpha_x\\
			\\
			&   2\sum_{s}\hspace{-0.05cm} \gamma(s,i,x)\hspace{-0.05cm} -\hspace{-0.05cm}  \sum_{s,i}\gamma(s,i,x)=0, \forall i \in \cI,   x\geq 0 :     \nu^i_x  \\
			&  \sum_{s,i,x}\gamma(s,i,x)-1=0 :\ \psi \\
			&   -\gamma(s,i,x)\leq 0,\  \forall s \in \{0,1\},  i \in \cI,  x \geq 0: \  \beta^i_{s,x}
			%	\\& \hspace{5cm} \rightarrow \quad \beta^i_{s,x},
			\label{objconst-2}
			%	\label{opt-lin-2}
		\end{align}
		where $\epsilon_1\geq 0$, $\epsilon_2\geq 0$, $\epsilon_3\geq 0$, $\epsilon_4\geq 0$, $\eta\geq 0$ and $\beta^i_{s,x}\geq 0$.
		By taking the derivative of the dual function with respect to $\gamma(1,1,x)$, $\gamma(1,2,x)$,  $\gamma(0,1,x)$, and $\gamma(0,2,x)$ for $x>0$, and also with respect to $t_0$ and setting them to zero, we have the following.
		\begin{subequations}
			\begin{align}
				&(\epsilon_2-2\epsilon_3+\eta)v(x)+\alpha_{x-1}-\lambda \alpha_x+\nu^1_x-\nu^2_x+\psi-\beta^1_{1,x}\nonumber \\& \hspace{6.5cm}=0\\
				&(-2+2\epsilon_1+\epsilon_2-\eta)v(x)+\alpha_{x-1}-\lambda \alpha_x+\nu^2_x-\nu^1_x\nonumber \\&\hspace{5cm} +\psi-\beta^2_{1,x}=0\\
				&(2\epsilon_1+\epsilon_2)v(x)+\alpha_{x-1}+\nu^1_x-\nu^2_x +\psi-\beta^1_{0,x}=0\\
				&(2\epsilon_1+\epsilon_2)v(x)+\alpha_{x-1}+\nu^2_x-\nu^1_x +\psi-\beta^2_{0,x}=0\\
				&-1+\epsilon_1+\epsilon_2+\epsilon_3+\epsilon_4=0. \label{epsilon}
			\end{align}
		\end{subequations}
		%From the above equations, a number of cases could arrise depending on the coefficients of $v(x)$ in each equation being zero or not.  We first assume that none of the coefficients are zero and therefore we can write
		Therefore, we have
		\begin{subequations}
			\begin{align}
				&v(x)=\frac{-\lambda \alpha_x+\alpha_{x-1}+\nu^1_x-\nu^2_x+\psi-\beta^1_{1,x}}{2\epsilon_3-\epsilon_2-\eta} \label{gama11}\\
				&=\frac{-\lambda \alpha_x+\alpha_{x-1}+\nu^2_x-\nu^1_x +\psi-\beta^2_{1,x}}{2-2\epsilon_1-\epsilon_2+\eta}\label{gama12} \\
				&=\frac{-\alpha_{x-1}-\nu^1_x+\nu^2_x -\psi+\beta^1_{0,x}}{2\epsilon_1+\epsilon_2}\label{gama01}\\
				&=\frac{-\alpha_{x-1}-\nu^2_x+\nu^1_x -\psi+\beta^2_{0,x}}{2\epsilon_1+\epsilon_2}.\label{gama02}
			\end{align}
			\label{gama}
		\end{subequations}
		Based on the above equations, we can have the following lemma.
		\begin{lemma}
			\label{equal}
			If there exists a $\tilde{x}>0$ for which we have $\gamma(0,1,\tilde{x})>0$, $\gamma(0,2,\tilde{x})>0$, $\gamma(1,1,\tilde{x})>0$, and $\gamma(1,2,\tilde{x})>0$, then $\eta=0$, $\epsilon_2=\epsilon_4=0$, and $\nu^1_x=\nu^2_x$ for all $x>0$. Further, we have $\beta^1_{0,x}=\beta^2_{0,x}$ and  $\beta^1_{1,x}=\beta^2_{1,x}$.
		\end{lemma}
		\begin{proof}
			Looking at equations \eqref{gama01} and \eqref{gama02}, we have $2\nu^1_x-2\nu^2_x +\beta^2_{0,x}=\beta^1_{0,x}$. Since $\gamma(0,1,\tilde{x})>0$ and $\gamma(0,2,\tilde{x})>0$ we have  $\beta^2_{0,\tilde{x}}=\beta^1_{0,\tilde{x}}=0$ and  it results in $\nu^1_{\tilde{x}}=\nu^2_{\tilde{x}}$. Also since $\gamma(1,1,\tilde{x})>0$ and $\gamma(1,2,\tilde{x})>0$, which results in $\beta^2_{1,\tilde{x}}=\beta^1_{1,\tilde{x}}=0$, we must have $2\epsilon_3-\epsilon_2-\eta=2-2\epsilon_1-\epsilon_2+\eta$ or equivalently, $\epsilon_1+\epsilon_3=1+\eta$. According to equation \eqref{epsilon}, we must have $\eta=0$ and $\epsilon_2=\epsilon_4=0$. Further, for any  $x>0$ and each type $i\in  \cI$, we either have $\gamma(1,i,x)>0$ or $\gamma(0,i,x)>0$ or both. Therefore, either $\beta^i_{1,\tilde{x}}=0$ or $\beta^i_{0,\tilde{x}}=0$ or both. Assume $\beta^1_{0,\tilde{x}}=0$. Then if $\beta^2_{0,\tilde{x}}=0$, we must have $\nu^1_{\tilde{x}}=\nu^2_{\tilde{x}}$ and the result is proved. If $\beta^2_{1,\tilde{x}}=0$, then from equations  \eqref{gama01} and \eqref{gama02} we must have $\nu^2_{\tilde{x}}\geq\nu^1_{\tilde{x}}$. From equations \eqref{gama11} and \eqref{gama12}, we must have  $\nu^1_{\tilde{x}}\geq\nu^2_{\tilde{x}}$. Therefore, $\nu^1_{\tilde{x}}=\nu^2_{\tilde{x}}$. Therefore, due to equation \eqref{gama}, we must have $\beta^1_{0,x}=\beta^2_{0,x}$ and  $\beta^1_{1,x}=\beta^2_{1,x}$ for all $x>0$.
		\end{proof}
		Note that all of the results of Lemma \ref{equal} hold if we have $2\epsilon_3-\epsilon_2-\eta=2-2\epsilon_1-\epsilon_2+\eta$, i.e., the denominators of the coefficients in equations \eqref{gama11} and \eqref{gama12} are equal. In other words, Lemma \ref{equal}  states a condition in which we must have $2\epsilon_3-\epsilon_2-\eta=2-2\epsilon_1-\epsilon_2+\eta$ and consequently, the rest of the results hold. However, these coefficients might be equal without having the condition stated in Lemma \ref{equal}. In the next lemma, we show what the solution looks like if we have $2\epsilon_3-\epsilon_2-\eta=2-2\epsilon_1-\epsilon_2+\eta$.
		\begin{lemma}
			\label{lem:eq}
			If we have $2\epsilon_3-\epsilon_2-\eta=2-2\epsilon_1-\epsilon_2+\eta$, then there is a threshold $\tilde{x}$ such that  for $x\geq \tilde{x}$ we have $\gamma(s=1,i,x)=0$  for all $i \in \cI$, and for $x<\tilde{x}$,   $\gamma(s=0,i,x)=0$ for all $i \in \cI$ except for some points $\tilde{\cX}=\{x_1, x_2, \ldots\}, \ x_k <\tilde{x} $ for which we can have $\gamma(0,1,x_k)>0$, or $\gamma(0,2,x_k)>0$, where  all $x_k \in \tilde{\cX}$ satisfy the following condition. There exists $\epsilon_1>0$ and $\psi$ such that
				\begin{align*}
					(2\sum_{x=0}^{x_k} \lambda^x v(x))\epsilon_1 + (\sum_{x=0}^{x_k} \lambda^x )\psi=\sum_{x=0}^{x_k-1} \lambda^x v(x), \ \forall x_k \in \tilde{\cX}.
				\end{align*}
			
		\end{lemma}
		\begin{proof}
			If we have $2\epsilon_3-\epsilon_2-\eta=2-2\epsilon_1-\epsilon_2+\eta$,  then $\eta=0$, $\epsilon_2=\epsilon_4=0$, and $\nu^1_x=\nu^2_x$ for all $x>0$. Further, we have $\beta^1_{0,x}=\beta^2_{0,x}$ and  $\beta^1_{1,x}=\beta^2_{1,x}$. This results in the following.
			\begin{subequations}
				\begin{align}
					&v(x)=\frac{-\lambda \alpha_x+\alpha_{x-1}+\psi-\beta^1_{1,x}}{2-2\epsilon_1} \label{gama11-eq}\\
					&=\frac{-\lambda \alpha_x+\alpha_{x-1} +\psi-\beta^2_{1,x}}{2-2\epsilon_1}\label{gama12-eq} \\
					&=\frac{-\alpha_{x-1} -\psi+\beta^1_{0,x}}{2\epsilon_1}=\frac{-\alpha_{x-1} -\psi+\beta^2_{0,x}}{2\epsilon_1}.\label{gama02-eq}
				\end{align}
				\label{gama-eq}
			\end{subequations}
			From the above equations, we can conclude that $\beta^1_{0,x}=\beta^2_{0,x}$ and $\beta^1_{1,x}=\beta^2_{1,x}$. It means that we if we have $\beta^1_{1,\tilde{x}}=\beta^2_{1,\tilde{x}}>0$ which results in $\gamma(s=1,i,\tilde{x})=0$ for all $i \in \cI$, the stationary distribution $\mu(x)$ is zero for all $x>\tilde{x}$ and therefore $\gamma(s=1,i,x)=0$ for all $i \in \cI$.   For $x<\tilde{x}$, we  have $\beta^1_{1,x}=\beta^2_{1,x}=0$.  We  also either have $\beta^1_{0,x}=\beta^2_{0,x}>0$, which results in $\gamma(s=0,i,x)=0$ for all $i \in \cI$ or we have $\beta^1_{0,x}=\beta^2_{0,x}=0$ which can allow us to have  $\gamma(0,1,x)>0$, $\gamma(0,2,x)>0$. Suppose either $\gamma(0,1,x)>0$ or $\gamma(0,2,x)>0$ for $x \in \{x_1, x_2, \ldots\}$. By writing the equation \eqref{gama-eq}  for $x=0$ and due to the fact that at least one of $\beta^1_{1,0}$ or $\beta^2_{1,0}$ is zero, we have $	v(0)=\frac{1}{2-2\epsilon_1}(-\lambda \alpha_0+\psi)$.
		
			For $x \in \{x_1, x_2, \ldots\}$ we have
			\begin{align}
				v(x)=\frac{-\lambda \alpha_x+\alpha_{x-1}+\psi}{2-2\epsilon_1}=\frac{-\alpha_{x-1} -\psi}{2\epsilon_1}
				\label{v0}
			\end{align}
			which results in
				$v(x)=-\frac{\lambda}{2}\alpha_{x}$.
			Using equations \eqref{gama-eq} and\eqref{v0} we have
			\begin{align}
				&\sum_{x=0}^{x_1} \lambda^x v(x)=\frac{1}{2-2\epsilon_1} (-\lambda^{x_1+1}\alpha_{x_1}+\psi\sum_{x=0}^{x_1} \lambda^x )\\
				& \Rightarrow \ (2\sum_{x=0}^{x_1} \lambda^x v(x))\epsilon_1 + (\sum_{x=0}^{x_1} \lambda^x )\psi=\sum_{x=0}^{x_1-1} \lambda^x v(x). \label{epsilon_psi}
			\end{align}
			having $x_1$, the above is an equation with respect to $\epsilon$ and $\psi$. In general, if we have $x_1$ and $x_2$, we should be able to determine $\epsilon$ and $\psi$ and we can have $x_k$ for $k\geq 3$ only if they result in linearly dependent equations in \eqref{epsilon_psi} and this might not be true for general $v(\cdot)$.
			%and $\gamma(s=0,j,\tilde{x})>0$ for any $i,j \in \{1,2\}$ and therefore, $\beta^i_{1,\tilde{x}}=0$ and $\beta^j_{0,\tilde{x}}=0$, we must have
			%	\begin{align}
			%      \frac{1}{2-2\epsilon_1}(\alpha_{x-1}-\lambda \alpha_x +\psi-\beta^2_{1,x})=\frac{1}{2\epsilon_1}(-\alpha_{x-1} -\psi+\beta^1_{0,x})
			%	\end{align}
			
			Note that if for all $x<\tilde{x}$ we have $\gamma(s=0,i,x)=0$ for all $i \in \cI$, then $q(1)=q(2)$ and the objective of the information designer would be zero. Therefore, this is probably not the solution of the optimization problem. In order to create discrimination between users of type 1 and type 2, the designer can only consider different policies for these two types  at $x_k \in \tilde{\cX}$.
		\end{proof}
		In Lemma \ref{lem:eq}, we investigated the solution under the assumption of  $2\epsilon_3-\epsilon_2-\eta=2-2\epsilon_1-\epsilon_2+\eta$. Note that due to equation \eqref{epsilon}, we can not have  $2\epsilon_3-\epsilon_2-\eta>2-2\epsilon_1-\epsilon_2+\eta$. Therefore if the equality does not hold, we have  $2\epsilon_3-\epsilon_2-\eta<2-2\epsilon_1-\epsilon_2+\eta$. In the next lemma we present some results under this inequality assumption.
		\begin{lemma}
			If  $2\epsilon_3-\epsilon_2-\eta<2-2\epsilon_1-\epsilon_2+\eta$, then we have the following.
		 If $v(x)> 0$ and $\sigma(s=1|1,x)>0$, we have $\sigma(s=1|2,x)=1$.
				Furthermore,  if $v(x)<0$ and $\sigma(s=1|2,x)>0$,  we have $\sigma(s=1|1,x)=1$.
		\end{lemma}
		\begin{IEEEproof}
			Looking at equation \eqref{gama}, if $v(x)> 0$, since $2\epsilon_3-\epsilon_2-\eta<2-2\epsilon_1-\epsilon_2+\eta$, we must have the following.
			\begin{align}
				v(x)&=\frac{\alpha_{x-1}-\lambda \alpha_x+\nu^1_x-\nu^2_x+\psi-\beta^1_{1,x}}{2\epsilon_3-\epsilon_2-\eta}\\
				&=\frac{\alpha_{x-1}-\lambda \alpha_x+\nu^2_x-\nu^1_x +\psi-\beta^2_{1,x}}{2-2\epsilon_1-\epsilon_2+\eta}\\
				& \Rightarrow \ \alpha_{x-1}-\lambda \alpha_x+\nu^1_x-\nu^2_x+\psi-\beta^1_{1,x}\\& \quad \quad <\alpha_{x-1}-\lambda \alpha_x+\nu^2_x-\nu^1_x +\psi-\beta^2_{1,x}\\
				&  \Rightarrow \ \beta^2_{1,x}+2\nu^1_x-2\nu^2_x<\beta^1_{1,x}
			\end{align}
			We also have the following.
			\begin{align}
				&-\nu^1_x+\nu^2_x+\beta^1_{0,x}=-\nu^2_x+\nu^1_x +\beta^2_{0,x}\\
				&\Rightarrow\ \beta^1_{0,x}=\beta^2_{0,x}+2\nu^1_x -2\nu^2_x \label{eq0}
			\end{align}
			We can have three cases for $\nu^1_x-\nu^2_x$. We either have $\nu^1_x> \nu^2_x$, which results in $\beta^2_{1,x}<\beta^1_{1,x}$, which means that $\beta^1_{1,x}>0$ and therefore, $\gamma(s=1,1,x)=0$. On the other hand, we must have  $\beta^1_{0,x}>0$ and therefore, $\gamma(s=0,1,x)=0$. This is a contradiction for those $x$'s that $\mu(x)>0$, which are the ones that we are interested in. If we have $\nu^1_x< \nu^2_x$, we must have  $\beta^2_{0,x}>0$, and therefore,  $\gamma(s=0,2,x)=0$. If we have $\nu^1_x=\nu^2_x$, we must have $\beta^1_{1,x}>0$ and therefore, $\gamma(s=1,1,x)=0$. Therefore, for $v(x)> 0$, we have $\sigma(s=1|2,x)=1$ if $\sigma(s=1|1,x)>0$.
			If $v(x)<0$, then we must have
			\begin{align}
				& \alpha_{x-1}-\lambda \alpha_x+\nu^1_x-\nu^2_x+\psi-\beta^1_{1,x}\\& \quad \quad >\alpha_{x-1}-\lambda \alpha_x+\nu^2_x-\nu^1_x +\psi-\beta^2_{1,x}\\
				&  \Rightarrow \ \beta^1_{1,x}+2\nu^2_x-2\nu^1_x<\beta^2_{1,x}.
			\end{align}
			Consequently, if $\nu^2_x>\nu^1_x$, we have $\beta^2_{1,x}>0$ and therefore, $\gamma(s=1,2,x)=0$. On the other hand, due to equation \eqref{eq0}, we have  $\beta^2_{0,x}>0$ and therefore, $\gamma(s=0,2,x)=0$, which is a contradiction. Hence, $\nu^2_x\leq\nu^1_x$. If $\nu^2_x<\nu^1_x$,  we have $\beta^1_{0,x}>0$ and so $\gamma(s=0,1,x)=0$. If $\nu^2_x=\nu^1_x$, we have $\beta^2_{1,x}>0$ and therefore, $\gamma(s=1,2,x)=0$. Hence, for $v(x)<0$ we have $\sigma(s=1|1,x)=1$ if $\sigma(s=1|2,x)>0$.
		\end{IEEEproof}}
%	\end{IEEEproof}

	\bibliographystyle{IEEEtran}
 %\bibliography{Nasimeh,achilleas18abrv,achilleas19_own,achilleas18_control}
\optv{submission}{	\input{root_bib.bbl}}
\optv{arxiv}{	\input{root_bib_arxiv.bbl}}

\end{document}

%% file: root_bib_arxiv.bbl
% Generated by IEEEtran.bst, version: 1.14 (2015/08/26)